\documentclass[A4paper,11pt]{article}

\usepackage{latexsym}
\usepackage{epsfig}
\usepackage{cite}
\usepackage{amssymb}
\usepackage{amsmath}
\usepackage{psfrag}
\usepackage{verbatim}

\newcommand{\ignore}[1]{}
\newtheorem{theorem}{Theorem}[section]

\newtheorem{lemma}[theorem]{Lemma}

\newtheorem{claim}[theorem]{Claim}

\newtheorem{definition}[theorem]{Definition}

\newcommand{\poly}{\hbox{poly}}
\newcommand{\otilde}{\widetilde{O}}
\newcommand{\qed}{\hfill $\Box$ \\}
\newcommand{\eps}{\varepsilon}

\newcommand{\EX}{\hbox{\bf E}}
\newcommand{\E}{\hbox{\bf E}}
\newcommand{\pr}{\hbox{\bf Pr}}

\title{Testing cycle-freeness: Finding a certificate}
\date{}

\author{
C. Seshadhri \\
IBM Almaden Research Center
}
\begin{document} 

\maketitle

\begin{abstract}
We deal with the problem of designing one-sided error property
testers for cycle-freeness in bounded degree graphs.
Such a property tester always accepts forests. Furthermore,
when it rejects an input, it provides a short cycle as a certificate.
The problem of testing cycle-freeness in this model
was first considered by Goldreich and Ron~\cite{GR97}.
They give a constant time tester with two-sided error (it
does not provide certificates for rejection) and prove a $\Omega(\sqrt{n})$ lower
bound for testers with one-sided error. We design a property
tester with one-sided error whose running time matches this lower bound
(upto polylogarithmic factors). Interestingly, this has
connections to a recent conjecture of Benjamini, Schramm,
and Shapira~\cite{BSS08}. The property of cycle-freeness
is closed under the operation of taking minors. This is the
first example of such a property that has an almost
optimal $\otilde(\sqrt{n})$-time
one-sided error tester, but has a constant time two-sided error tester.
It was conjectured in~\cite{BSS08} that this happens for a vast class of
minor-closed properties, and this result can seen as the first indication
towards that.

\end{abstract}

\section{Introduction}

Given a massive data set, one may wish to learn some properties
of this data set \emph{without} reading the whole input. The
standard decision problem is to accept an input if it satisfies
some given property and reject it otherwise. In property
testing~\cite{RS96, GGR98}, we deal with a relaxed version
of this decision problem: accept the input if it satisfies the
property, and reject if the input is \emph{far} from
satisfying the property. There are algorithms that can perform
this relaxed version, for a vast array of properties, in
\emph{sublinear} (or even independent of the input size) time
(see surveys~\cite{CSSurvey,F01,G98,R01}).
Goldreich, Goldwasser, and Ron first introduced property
testing for dense graphs in~\cite{GGR98}, and gave
testers for many different properties. Deep theorems
about the class of properties that can tested
in this model
have been given in~\cite{AS05,AENS06}.

In this paper, we will deal with property testing
in bounded degree graphs, as introduced by
Goldreich and Ron~\cite{GR97}. The input
undirected graph $G$ has degree bound $d$ (considered
as a constant) and $n$ vertices. It is represented by an adjacency list,
and therefore has a representation size of $nd$.
We will assume that $G$ is simple without self-loops.
Henceforth, we will only consider graphs of this form.
For every vertex $v$, we have access to the at most $d$ neighbors
of $v$. This allows a testing algorithm to perform walks in $G$,
something that cannot be done with adjacency matrices.
Given a property $\cal P$ and distance parameter $0 < \eps < 1$,
the graph $G$ is $\eps$-far from $\cal P$ if $G$ must
be changed at more than $\eps nd$ edges to have property $\cal P$.
A property tester accepts $G$ if it has property $\cal P$
and rejects if $G$ is $\eps$-far from $\cal P$ (both hold with
high probability). The running time of the property tester should be
$\poly(d\eps^{-1}) o(n)$ (if there is no dependence on $n$,
then we say that the running time is constant).
A stronger tester, called a \emph{tester with one-sided
error} (for ease of notation, we will call this
a one-sided tester) accepts with probability $1$ if $G$ has $\cal P$.
Also, if it rejects $G$, it provided a sublinear-sized
certificate that $G$ does not have $\cal P$.
This makes one-sided testers far more interesting, and harder
to design. For one-sided testers, we have to 
prove that if $G$ is $\eps$-far from
having $\cal P$, then a sublinear sized certificate that $G$ is not in $\cal P$
can be found in sublinear time.

Goldreich and Ron~\cite{GR97} showed that many connectivity
properties are testable in time independent of $n$.
Progress towards finding general classes of testable properties
for bounded degree graphs was first made
by Czumaj, Sohler, and Shapira~\cite{CS07,CSS}. Benjamini, Schramm,
and Shapira~\cite{BSS08} made a breakthrough towards
these characterizations, and showed that any minor-closed property
is testable. These algorithms run in time independent of $n$
but are not one-sided testers. Many interesting properties
do not fall into this framework. Indeed, Goldreich and Ron~\cite{GR97}
gave a lower bound of $\Omega(\sqrt{n})$ for property
testing bipartiteness in the bounded degree model.

There has been a surprising dearth of knowledge about
properties that require more than constant time
for testing. Goldreich and Ron, in a very involved
result, give a tester for bipartiteness~\cite{GR99}
whose running time almost matches\footnote{upto polylogarithmic factors} their lower bound of $\sqrt{n}$.
One of the very nice features of this work was
the first use of random walks for property testing.
The property of expansion received a lot
of attention recently, and also has almost optimal
property testers running in time $\otilde(\sqrt{n})$
\cite{GR00,CS-F07,KS,NS}.

\subsection{Cycle-freeness}

It is very intriguing to look at the natural and basic property
of \emph{cycle-freeness}, where the tester accepts
if the input is a forest, and rejects if it is
$\eps$-far from being one. Goldreich and Ron in
their initial paper~\cite{GR97} give a constant time tester
for this property. Marko and Ron extend
this to estimate the distance to cycle-freeness
in constant time~\cite{MR06}.
The tester of~\cite{GR97} is \emph{not} 
one-sided\footnote{This tester can reject forests
and does not provide a certificate on rejection.},
and fundamentally rests on the fact that
in a tree, the number of edges is less than
the number of vertices. They also prove
the somewhat surprising lower bound that one-sided 
cycle-freeness testing cannot be done in
$o(\sqrt{n})$ time. This is the problem
that we shall attack. We resolve
the problem of testing cycle-freeness by giving
a one-sided property tester 
whose running time is almost optimal.

\begin{theorem} \label{thm-main} There exists a one-sided property
tester for cycle-freeness whose running time is $\sqrt{n} \ \poly(d \eps^{-1}\log n)$.
In other words, this property tester always accepts the input $G$,
if $G$ is a forest. If $G$ is $\eps$-far from being a forest,
it rejects $G$ with probability at least $2/3$.
Moreover, when the tester rejects $G$, it provides
a certificate in the form of a short cycle (of length $\poly(\eps^{-1}\log n)$).
\end{theorem}

Why is this challenging? The difficulty comes from the
fact that the property tester needs to \emph{find} a cycle
in a graph that $\eps$-far from being cycle-free.
Most property testers for bounded degree graphs
essentially sample some random vertices and looks
at small neighborhoods around them. This
cannot work for a one-sided cycle-freeness tester.
The only testers that use superconstant time (actually $\otilde(\sqrt{n})$
time) are the testers for bipartiteness and expansion.
The main tool they use is that of random walks.
It is very interesting to have another property
that can be tested through a non-trivial use of 
random walks. Combinatorially, a one-sided cycle-freeness
tester is more interesting because when it rejects
the input $G$, it provides a certificate
in the form of a short cycle.

We use many ideas from the tester of Goldreich and Ron for
bipartiteness~\cite{GR99}.
One of their major contributions is to provide a framework to analyse
short random walks in bounded degree graphs, and
argue how it suffices to deal with ``well-connected"
graphs. Much of their result dealt with proving
theorems about this framework. With this in place,
they only need to test bipartiteness within these
special subgraphs. It turns out that testing
cycle-freeness within such subgraphs is much harder,
and requires new ideas.

We develop combinatorial tools to study the behavior
of random walks in graphs that are almost forests.
We show how random walks in such graphs can be
used to detect cycles. A major challenge is
to show that this can done in $\otilde(\sqrt{n})$ time, to
match the lower bound of Goldreich and Ron.
We are able to give a characterization of some special
edges in terms of random walk behavior, such
that these special edges form a forest. We then argue
that the presence of many non-special edges allows
our property tester to find a cycle.

Benjamini, Schramm, and Shapira~\cite{BSS08} also
observe the gap between property testers
and their one-sided cousins. We first explain
graph minors.
A minor of a graph $G$ is obtained by removing vertices,
edges, and contracting edges. A graph $G$ is $H$-minor free,
if $H$ is not a minor of $G$. A property is minor-closed
if the property is preserved on taking minors.
They show
that every minor-closed property is testable (with two-sided error)
with a constant number of queries and \emph{conjecture}
that all such properties can be tested with one-sided
error in $\otilde(\sqrt{n})$ time\footnote{Based on lower
bounds in~\cite{GR97}, they conclude that for a large class
of minor-closed properties, there is a lower bound
of $\Omega(\sqrt{n})$ for one-sided testing.}. 
We provide
the first such property with exactly this behavior,
thereby giving some (very small) evidence for this
conjecture. The property of being cycle-free is certainly
minor-closed. There is a constant time tester
for cycle-freeness (with two-sided error) and we give
an almost optimal $\otilde(\sqrt{n})$-time one-sided tester.
It is possible that the tools developed here could
yield some insight into the conjecture of~\cite{BSS08}.

In the language of minors, we can think of cycle-free graphs 
as $K_3$-minor free graphs, where $K_3$ is the triangle.
Our tester finds in sublinear time a forbidden minor in a graph
that is far from being $K_3$-minor free.
Our special edges form a $K_3$-minor free graph, and the remaining
edges boost the probability of finding a $K_3$-minor in sublinear
time. It would be very interesting to generalize our 
technique to other minors, making progress towards
the conjecture of~\cite{BSS08}. Specifically,
if we could do the same for $K_5$ and $K_{3,3}$ minors, 
then we would have a sublinear time procedure to find
forbidden minors in graphs far from being planar.

\section{Preliminaries and intuition}

The basic idea is to perform $O(\sqrt{n})$ random walks
of length $O(\log n)$ from a random subset of $O(\eps^{-1})$
vertices of $G$. We use a standard lazy random walk: each edge
of $G$ is assigned a probability of exactly $1/2d$. For a vertex
$v$, suppose that there are $d' \leq d$ edges incident to $v$.
There is a self-loop at $v$ with probability $1 - d/2d' \geq 1/2$.
We hope if $G$ is $\eps$-far from
being cycle-free, then we can find two \emph{different paths}
between two vertices. This provides us with a cycle
that can be output as a witness. Note that it is not necesarily
the case that there are an extremely large number of cycles
in such a $G$. Consider the case where $G$ consists
of $\eps n$ disjoint cycles of size $1/\eps$. The number
of cycles is relatively small, but on the other hand, 
they are not hard to detect. 

A technical lemma of~\cite{GR99} shows (roughly speaking)
that we can partition the vertices of 
$G$ into ``well-connected components".
Such a component $S$ has the following property: there
is a special vertex $s \in S$ such that short (polylogarithmic) random
walks from $s$ reach all vertices of $S$ with
high probability. Furthermore, there are few
edges connecting $S$ to the remaining vertices.
These well-connected components
do not necesarily have to look like expanders. Such
a component could even be a very small (polylogarithmic) path.
It is merely the case that random walks of a specfied length
from $s$ spread out to all of $S$.
One of the nice features of this lemma is that we do not
have to explicitly compute this partition. 
Note that the random walks
are performed in $G$, not in the induced subgraph of $S$. 
It is some kind of a thought experiment that allows us to imagine
that we are performing walks inside well-connected
components. The testing algorithm is completely oblivious to this
partition, and simply peforms walks in $G$.

We can therefore focus on any one
such well-connected component $S$
with $s$ as the special vertex.
For simplicity, assume that all edges
and vertices that we mention below are in $G_S$,
the induced subgraph on $S$.
Suppose we perform a large number of random walks
from $s$ and notice that there is an edge $e$
that always lies on a path from $s$ to $v$.
The edge $e$ is a sort of bottleneck.
Roughly speaking, we will show that there cannot be too
many bottleneck edges since they all form a forest.
Now, if $G_S$ was far
being a forest, then there must be many
edges in $G_S$ that are not bottlenecks.
We then show that the
presence of many non-bottleneck edges
will boost the probability of detecting
a cycle.

Now, we need to prove that a relatively
small number of walks can find a cycle. 
This requires a delicate argument.
We are able to argue that the probability
of two random walks witnessing a cycle is around $1/n$.
So, using a birthday paradox like argument,
(as done in~\cite{GR99}), we hope that it suffices
to perform $O(\sqrt{n})$ to detect a cycle.
Let the random variable $X$ be an indicator
for two random walks witnessing a cycle.
It boils down to seeing the behavior of
$X$. We showed that the expectation
of $X$ behaves like $1/n$, but the variance
can be too large. We look into \emph{when}
this variance is too large. This
happens because there a relatively small
number of paths which are part of most
cycles. It turns out that we can handle this
case separately (without dealing
with bottleneck edges) in exactly $O(\sqrt{n})$ time.

\section{Within a well-connected subgraph}

We first perform an
analysis in a well-connected component of $G$.
This component will be referred to as $S$, and
the induced subgraph on it as $G_S$.
We fix some parameters: the length of the random walk\footnote{The exact polynomial factor will be decided later.} 
$\ell = \poly(\eps^{-1}\log n)$,
a probability $\frac{\eps}{\sqrt{|S|n}\log n} < \alpha < 1$, and $m = c\eps^{-3}\sqrt{n}\ell\log^2n$
($c$ is a sufficiently large constant).
Let $q_v$ be the probability that a random 
walk from $s$ \emph{reaches} $v$. This means
that the walk from $s$ encounters $v$ during
its course.

Let $s \in S$ be a \emph{special vertex}
in $S$ if $q_v > \alpha$ for all $v \in S$.
The analysis in this section will only deal with
the vertices in $S$, although the random walk is 
being performed in $G$. 
From now on, the term random walk will refer only to walks from $s$
of length $\ell$.

Suppose that $G_S$ is far from being cycle-free. Our aim is to show that when many
random walks are performed from $s$, a cycle is detected with
high probability. A walk is
a sequence of edges with possible backtracking and self-loops.
The \emph{path} induced by this walk is the sequence of edges
from $s$ to $v$ obtained by removing self-loops and retraced
paths. If we can find two different paths from $s$ to $v$, then
this is a cycle in $G$.
Note that two different walks can induce the same path, so finding
two different walks from $s$ to $v$ yields no information.
This is a major problem when trying to find a cycle.
Just because there is a high probability of reaching $v$
from $s$ does not guarantee that different walks from $s$
to $v$ form a cycle. It could be possible that most
walks from $s$ to $v$ induce the same path.

When we deal with a walk from $s$ that reaches $v$,
we will not deal with the portion of the walk \emph{after}
$v$. In other words, we will say that two walks reach
$v$ by the same path if the walks induce the same
path upto $v$ (they can diverge after that).

We now describe a process {\sc Cycle Finder} that attempts
to find a cycle given such a vertex. This is analogous
to the algorithm of~\cite{GR97} for bipartiteness. 
We remind the reader that a graph $G$ is $\eps$-far
from being cycle-free, if at least an $\eps$-fraction
of edges need to be removed from $G$ to eliminate
all cycles. The main lemma
of this section is:

\begin{lemma} \label{lem-special} Let $s \in S$ be a special vertex for $S$ and $G_S$ be $\eps$-far from being cycle-free.
Then {\sc Cycle Finder} with $s$ as input will detect a cycle with probability
at least $2/3$.
\end{lemma}

\begin{center}
\fbox{\begin{minipage}{\columnwidth} {\sc Cycle Finder}\\
{\bf Input:} Vertex $s \in S$
\begin{enumerate}
    \item Perform $m$ random walks of length $\ell$ from $s$.

		\item Look at the subgraph $G'$ induced by all the vertices reached
		by the random walks. ($G'$ has at most $dm\ell$ edges and is therefore 
		of $\otilde(\sqrt{n})$ size.) If $G'$ is
		cycle-free, ACCEPT. Otherwise REJECT and output a cycle.

\end{enumerate}
\end{minipage}}
\end{center}

Before we study the behavior of this procedure, we will need a few definitions.
We will look at every edge $(u,v)$ as a pair of directed edges $\langle u,v \rangle$
and $\langle v,u \rangle$. For the sake of clarity, angular braces 
represent a directed edge and parentheses denote an undirected edge.

\begin{definition}\label{def-iso} 
A vertex $v \in G_S$ is \emph{isolated} if more than half of the walks that reach $v$ induce the same path (called
the \emph{dominant path}). Furthermore, the probability
of a walk reaching $v$ and inducing a different path is $< \alpha/2$.
\end{definition}

\begin{definition} \label{def-dominant} A directed edge $\langle u,v\rangle$ ($(u,v)$ is in $G_S$) 
is \emph{dominant} if the following conditions hold: 

\begin{enumerate} 
	\item The vertex $v$ is isolated. 
	\item The probability that a walk reaches $v$ without passing through the edge $(u,v)$ is $< \alpha/2$.
\end{enumerate}

An undirected edge $(u,v)$ is \emph{recessive}, if neither $\langle u,v \rangle$
or $\langle v,u \rangle$ are dominant.
\end{definition}

The dominant edges are the bottlenecks that we mentioned in the previous
section. We now prove some small claims about these edges that will allow
us to bound their number.

\begin{claim} \label{clm-dom} Given a dominant edge $\langle u,v\rangle$, the dominant
path to $v$ passes through $(u,v)$.
\end{claim}

\begin{proof} Assume by contradiction that the dominant
path does not contain the edge $(u,v)$.
By the definition of the dominant path,
the probability that a walk reaches $v$ by a path
passing through $(u,v)$ is $< \alpha/2$.
By the properties of $S$, we know that the probability
of reaching $v$ is $\geq \alpha$. This implies
that the probability of reaching $v$ by a path
that does not contain $(u,v)$ is $> \alpha/2$.
This violates the second part of Definition~\ref{def-dominant}.
\qed
\end{proof}

\begin{claim} \label{clm-both} Both $\langle u,v\rangle$ and $\langle v,u\rangle$ cannot be dominant.
\end{claim}

\begin{proof} Since $\langle u,v \rangle$ is dominant, the probability that a walk reaches $v$ without
going through $(u,v)$ is $< \alpha/2$. The probability of reaching $v$ is $\geq \alpha$.
Therefore, the probability of a walk reaching $v$ after taking $(u,v)$ is $> \alpha/2$.
Note that the portion of these walk from $s$ to $v$ cannot contain any other occurrence
of $v$. Therefore, these walks reach $u$ without passing through $(u,v)$, contradicting that $\langle v,u\rangle$
is dominant.
\qed
\end{proof}

With these claims in hand, we are ready to connect dominant edges
to cycle-freeness. 

\begin{lemma} \label{lem-forest} The set of dominant edges form a directed forest.
\end{lemma}

\begin{proof} Consider the subgraph (directed) $U$ created by the dominant
edges.
First we show that no vertex in $U$ can have indegree $\geq 2$. Suppose
there are edges $\langle u,v\rangle$ and $\langle w,v\rangle$. With probability
$> \alpha/2$, the path reaching $v$ goes through edge $(u,v)$ (or $(w,v)$). Note that
the probability of reaching $v$ through a dominant path is $> \alpha/2$. Since this
must both pass through $(u,v)$ and $(w,v)$, we reach a contradiction. 

Suppose the graph $U$ has an undirected cycle. By what was proven above, this
must also be a directed cycle. Consider the directed cycle $\langle u_0, u_1\rangle$,
$\langle u_1, u_2\rangle, \cdots$, $\langle u_{k-1},u_k\rangle$, $\langle u_k, u_0\rangle$.
Let us look at all walks that reach $u_k$ through the dominant path (which passes
through $u_{k-1}$). These walks have a cumulative probability of $> \alpha/2$.
Suppose this path reaches the set of vertices $u_0, \cdots, u_k$ first at
$u_i$ ($i \neq 0$). Since this path does not pass through $u_{i-1}$, it
must be the case the with probability $> \alpha/2$, a walk reaches $u_i$
without passing through the edge $(u_{i-1},u_i)$. This contradicts the
fact that $\langle u_{i-1}, u_i\rangle$ is dominant.

With probability $> \alpha/2$, a walk reaches $u_k$ inducing a 
path (which is dominant) that passes through $u_0$. Since $\langle u_k, u_0\rangle$ is dominant,
with probability $> \alpha/2$, a walk reaching $u_0$ induces a fixed
path passing through
$(u_k, u_0)$. Since $\langle u_{k-1},u_k\rangle$ is dominant,
this fixed path must follow the dominant path to $u_k$. But we
showed that the dominant path to $u_k$ passes through $u_0$.
This is absurd since it contradicts the fact that
the dominant path to $u_0$ passes through $u_k$. This proves that
$U$ is a directed forest.
\qed
\end{proof}
 
\subsection{Analysis}
 
Assume that $G_S$ is $\eps$-far from being cycle free.
By Lemma~\ref{lem-forest}, there must be at least $\eps nd$
recessive edges. We will show that the presence of many recessive edges 
will imply that {\sc Cycle Tester} will find cycles 
with high probability. Let $X$ be the indicator random variable
for two independent random walks performed by {\sc Cycle Tester}
forming a cycle.
The approach is to show that the expectation of $X$
is large enough (around $1/n$) because of the
recessive edges.
To ensure that a sublinear number of walks suffice
to find a cycle, we need the variance of $X$ to be small. It
turns out that the variance of $X$ \emph{can} be quite
large, and this is because of certain paths
that we call~\emph{heavy}. Fortunately for us,
cycles involving heavy paths can actually
be detected very quickly.

We split our analysis
into two parts. If there are many heavy paths,
then we can argue directly that there is
a high probability of finding a cycle. On the
other hand, if there are not too many heavy
paths, we can contain that variance of $X$.
In either case, we only need $\otilde(\sqrt{n})$ 
time to find cycles.

\begin{definition} For a random walk $W$, let $cyc_W$ be the probability
that a random walk of length $\ell$ from $s$ creates a cycle with $W$.

If $cyc_W > 1/\sqrt{n}$, then $P$ is \emph{heavy}. Otherwise, it is \emph{light}.
\end{definition}

We first deal with heavy paths. 
We split the probability mass of walks $q_v$ into two disjoint parts: $q^L_v$ ($q^H_v$)
is the probability that $v$ is reached by a light path (resp. heavy path).
Obviously, $q_v = q^L_v + q^H_v$. Call vertex $v$ \emph{blue} if $q^H_v > \alpha/4$.

\begin{claim} \label{clm-blue} If there are more than $\eps |S|/2$ blue vertices, then
{\sc Cycle Tester} finds a cycle with probability $> 2/3$.
\end{claim}

\begin{proof} Consider a random walk $W$ of length $\ell$ from $s$.
Let $X_W$ be the indicator random variable for the event that 
the path $P$ induced by $W$ is heavy. Let indicator random variable $X_v$ 
be $1$ iff $P$ is heavy and passes through $v$. Since $W$ passes through
at most $\ell$ vertices, $X_W \geq \sum_v X_v/\ell$. By linearity of expectation,
$\E[X_W] \geq \sum_v \E[X_v]/\ell = \sum_v q^H_v / \ell$.
By the claim assumption, there are at least $\eps |S|/2$ vertices $v$
such that $q^H_v > \alpha/4$. That implies that 
$\E[X_W] \geq \eps \alpha |S| / 8\ell = \Omega(\frac{\eps \sqrt{|S|}}{\sqrt{n}\ell \log n})$.
The expectation $\E[X_W]$ is the probability that the random walk $W$ induces a heavy path.
Consider the first $m/2$ walks performed by {\sc Cycle Tester}. Since $m > c\eps^{-1}\sqrt{n}\ell\log n$,
with probability at least $5/6$, one of the first $m/2$ walks will induce
a heavy path (abusing notation, call this $W$).

We know that $cyc_W > 1/\sqrt{n}$. Since $m > c\sqrt{n}$, with probability at least $5/6$, the latter $m/2$ walks
will contain a walk that creates a cycle with $W$. Therefore,
by a union bound, with probability at least $2/3$, {\sc Cycle Tester} finds
a cycle.
\qed
\end{proof}

We will henceforth assume that there are at most $\eps |S|/2$ blue vertices.
Call an edge non-blue if neither of its endpoints are blue. Since
the degree bound is $d$, we have at least $\eps|S|d/2$ non-blue recessive
edges. Consider two independent
random walks performed from $s$. Let ${\cal E}_v$
be the event that the two random walks are light and find a cycle in 
one of the following ways:
\begin{enumerate}
	\item The two walks reach $v$, and the corresponding paths induces
a cycle (which may not contain $v$; it is even possible
that one walk by itself creates a cycle).
	\item One walk reaches $v$, and the other walk reaches a neighbor $u \in S$.
The two walks and the edge $(u,v)$ induce a cycle.
\end{enumerate}
When the event ${\cal E} = \bigcup_v {\cal E}_v$ happens, then two walks performed by {\sc Cycle Tester} 
detect a cycle. Goldreich and Ron~\cite{GR99} use similar events in their bipartiteness
tester.

\begin{lemma} \label{lem-tester} If there are less than $\eps|S|/2$ blue
vertices, $\pr({\cal E}) \geq \eps |S| \alpha^2/16\ell$.
\end{lemma}

\begin{proof} We focus on two walks. There are at least $\eps|S|d/2$ non-blue recessive
edges. Consider such an edge $(u,v)$. Both $q^L_u, q^L_v$
are at least $3\alpha/4$. Neither $\langle u,v \rangle$ or 
$\langle v,u \rangle$ are dominant. We will
split up into cases depending of which part of Definition~\ref{def-dominant}
fails.

{\bf Case 1:} For both $\langle u,v \rangle$ and $\langle v,u \rangle$, part 2 of
Definition~\ref{def-dominant} fails. Therefore, the probability that a walk
reaches $u$ (and $v$) without passing through $(u,v)$ is at least $\alpha/2$.
By a union bound, the probability that such a walk is light is at least $\alpha/4$.
Therefore, the probability of ${\cal E}_u$ (and ${\cal E}_v$) is at least $\alpha^2/16$.

{\bf Case 2:} Wlog, let $v$ not be isolated. Now, suppose further that there is
no dominant path to $v$. The probability that the first walk is light and reaches $v$ is at least
$3\alpha/4$, by the properties of $S$. The probability that the second walk reaches
$v$ by a different path is at least $\alpha/2$, since no dominant path exists.
The probability that this walk is light is at least $\alpha/4$.
Therefore, the probability of ${\cal E}_v$ is at least $3\alpha^2/16$. Now, suppose
a dominant path (it must be light) to $v$ exists. Then the probability that a light walk reaches $v$ but
induces a different path (from the dominant one) is at least $\alpha/4$. We
can see that the probability of ${\cal E}_v$ is at least $\alpha^2/8$.\\

We can conclude that if $(u,v)$ is a non-blue recessive edge, then either $\pr({\cal E}_u)$ or $\pr({\cal E}_v)$
is at least $\alpha^2/16$. Therefore, 
there are at least $\eps |S|$ vertices $v$ such that $\pr({\cal E}_v) \geq \alpha^2/4$.
Fix a pair of walks. How many different ${\cal E}_v$'s can this pair belong to?
This number is at most $2\ell$, since that is the total number of vertices
reached by these walks. Therefore, 
$$ \pr({\cal E}) \geq \sum_v \pr({\cal E}_v) / 2\ell \geq \eps|S|\alpha^2/16\ell$$
\qed
\end{proof}

We now look at the probability that {\sc Cycle Tester} finds a cycle. Since {\sc Cycle Tester}
performs $m$ walks, we hope to boost the probability of rejection.

\begin{lemma} \label{lem-cheb} If there are less than $\eps |S|/2$ blue vertices, 
then {\sc Cycle Tester} finds a cycle with probability at least $2/3$.
\end{lemma}

\begin{proof} There are $m$ walks performed. Let ${\cal F}_{ij}$ be the event
that for the $i$th and $j$th walks, the event ${\cal E}$ happens (a cycle
is detected in the way described above). Let $X_{ij}$ be the indicator
random variable of this event. Obviously, the number of cycles
detected is at least $X = \sum_{i,j} X_{ij}$. We have that
$\E[X_{ij}] = 2\pr({\cal E})$ (denoted by $\mu$)\footnote{The factor of $2$
comes about because in $X_{ij}$, we have labelled walks.}
 and $\E[X] = {m\choose 2} \mu$.
It will be convenient to work with random variables $\overline{X_{ij}} = X_{ij} - \mu$.
We can see that $\E[\overline{X_{ij}}] = 0$ and 
$$\E[\overline{X_{ij}}^2] = \E[X^2_{ij}] - 2\mu\E[X_{ij}] + \mu^2
= \mu - \mu^2 \leq \mu$$.
We now show that the variance of $X$ is not too large.
\begin{eqnarray*}
	var(X) & = & \EX[(X - \mu M)^2] \\
	& = & \EX[(\sum_{i,j}\overline{X}_{ij})^2] \\
	& \leq & \sum_{i,j} \EX[\overline{X}_{ij}^2]
	+ \sum_{\substack{(i,j), (i', j') \\ i \neq i', j \neq j'}} \EX[\overline{X}_{ij} \overline{X}_{i'j'}]
	+ 6\sum_{i < j < k} \EX[\overline{X}_{ij}\overline{X}_{ik}] \\
	& \leq & \mu m^2 + 0 + 6\sum_{i < j < k} \EX[\overline{X}_{ij}\overline{X}_{ik}]
\end{eqnarray*}

The second term is $0$ because for pairs $(i,j), (i', j')$ where $i \neq i', j \neq j'$,
$\E[\overline{X}_{ij}]$ and $\E[\overline{X}_{i'j'}]$ are independent. Dealing with
the last term is more difficult. This is where the properties of light walks
will help us bound this term. Fixing some $i,j,k$, we will deal with $\E[X_{ij}X_{ik}]$. 
When is $X_{ij}X_{ik}$ equal to $1$? This happens when both the $j$th and $k$th walks
form cycles with the $i$th walk (in the way prescribed by $\cal E$).
Given a light walk $W$, let $r_W$ denote the probability that a random walk
forms a cycle with $W$ as described by $\cal E$. In other words, given a random 
walk $W'$, $r_W$ is the probability that pair of walks $W$ and $W'$ are in $\cal E$.
Note that $r_W$ is at most $cyc_W$, which is bounded by $1/\sqrt{n}$.
\begin{eqnarray*}
	\E[X_{ij}X_{ik}] & = & \sum_{W : \ W \ \textrm{light}} \pr(W) r^2_W \\
	& \leq & (1/\sqrt{n}) \sum_{W : \ W \ \textrm{light}} \pr(W) r_W \\
	& = & (\mu/\sqrt{n})
\end{eqnarray*}
We can now bound the variance of $X$, noting that $m \geq \sqrt{n}$.
$$	var(X) \leq \mu m^2 + 6\mu m^3/\sqrt{n} \leq 7\mu m^3/\sqrt{n} $$
By Chebyschev's inequality,
$$ \pr[|X - \mu m^2| > k \sqrt{7\mu m^3n^{-1/2}}] < 1/k^2 $$
By Lemma~\ref{lem-tester} and the choice of parameters:
$$	\mu m \sqrt{n} = 2\pr({\cal E}) \times \frac{cn\ell\log^2n}{\eps^3} \\
\geq 2\eps |S| \times \frac{\eps^2}{16|S|n\ell\log^2n} \times \frac{cn\ell\log^2n}{\eps^3} > 28
$$
This implies that $\mu m^2 > 2\sqrt{7\mu m^3n^{-1/2}}$, and that
$X$ is zero with probability at most $1/4$.
\qed
\end{proof}

The proofs of Claim~\ref{clm-blue} and Lemma~\ref{lem-cheb} prove
Lemma~\ref{lem-special}.

\section{The general case}

We now use the lemmas proved in the previous section to give a property tester
for the general case. We essentially use a lemma of Goldreich and Ron~\cite{GR99}
to show how a tester for the a well-connected subgraph can
be used for a general graph. The proof given here parallels that of~\cite{GR99}
for bipartiteness. We slightly paraphrase for convenience:

\begin{lemma} \label{GR} \cite{GR99} Let $H$ be subgraph of $G$ with at least $\eps n/4$ vertices.
Let $\ell = (\log (n/\eps))^6\eps^{-8}$. For at least half of the 
vertices $s$ in $H$, the following is true: there is a subset of vertices
$T_s$ in $H$ such that, 
	\begin{enumerate}
		\item The number of edges from $T_s$ to the rest of $H$ is at most $\eps d|T_s|/2$.
		\item For every $v \in T_s$, $q_v \geq \frac{\eps}{\sqrt{|S||H|\log(n/\eps)}} > \frac{\eps}{\sqrt{|S|n}\log n}$.
	\end{enumerate}
\end{lemma}

\begin{center}
\fbox{\begin{minipage}{\columnwidth} {\sc Cycle-freeness Tester}\\
\textbf{Input:} Graph $G$, and $0 < \eps < 1$
\begin{enumerate}
	\item Choose a random subset $R$ of $c/\eps$ vertices ($c$ is sufficiently large constant).
	\item For every $s \in R$, run {\sc Cycle Finder} with $s$ as input. Output ACCEPT
	iff all these runs output ACCEPT.
\end{enumerate}
\end{minipage}}
\end{center}

The property tester, {\sc Cycle-freeness Tester} simply calls {\sc Cycle Finder} from a sufficiently large random
sample of vertices. We show that it rejects graphs far
from being cycle-free with high probability.

\begin{theorem} If $G$ is $\eps$-far from being cycle-free, then
the procedure {\sc Cycle-freeness tester} rejects with probability $> 2/3$.
\end{theorem}

\begin{proof} As in most property testing proofs, we will actually prove the contrapositive: if 
{\sc Cycle-freeness tester} rejects $G$ with probability less than $2/3$,
then $G$ is $\eps$-close to being cycle-free. This will be shown by removing
$\eps nd$ edges and getting a forest. The behavior of {\sc Cycle Finder}
on the random subset $R$ will guide us to the set of edges that need to be removed.
Technically, we will run {\sc Cycle-freeness tester} with input
$\eps/4$, instead of $\eps$.

Let vertex $s$ be called \emph{strong} if the probability that
{\sc Cycle Finder} (with $s$ as input) detects a cycle is at least $2/3$.
We observe that at most $\eps n/10$ vertices can be strong.
If not, then with probability $> 5/6$, the random subset $S$
contains at least $c/20$ strong vertices. The probability
that at least one of these strong vertices is rejected by {\sc Cycle Finder}
is $> 5/6$. This shows that {\sc Cycle-freeness Tester} rejects
$G$ with probability at least $2/3$. This contradicts our initial
assumption.

We will construct a special partition of sthe vertex set of $G$ into
$U_1, U_2,\cdots,U_r$. For $1 \leq i < r$,
the induced subgraph $G_{U_i}$ is $\eps/4$-close
to begin cycle-free. This is not
true for $U_r$, but the size $|U_r|$ is at most $\eps n/4$.
Furthermore, the number of edges from $U_i$ to $\bigcup_{j > i} U_j$
is at most $\eps|U_i|d/2$. If such a partition
exists, we can remove $\eps nd$ edges to make
$G$ cycle-free. First, we remove all edges incident
to $U_r$ and any edge that crosses this partition.
This is a total of $\eps nd/4 + \eps nd/2$ edges.
This leaves us with $G_{U_i}$, for $i < r$. In each
of these, we can remove at most $\eps |U_i|/4$ edges
to get a forest. The total number of edges
removed is at most $\eps nd$ edges.

We construct this partition inductively.
At any intermediate stage, we have a partition
$U_1,U_2,$ $\cdots,U_i$. Each $G_{U_j}$, $j < i$
is $\eps/4$-close to being cycle-free, and the
number of edges going out of each $U_i$ is small
(as given by the condition in previous paragraph).
Initially, $U_1$ is simply the complete vertex
set and previous condition is vacuously true.
If $|U_i| \leq \eps n/4$, then we are done.
Suppose that is not the case, so $|U_i| > \eps n/4$.
By our initial assumption, there are at
most $\eps n/10$ strong vertices.
These conditions together with Lemma~\ref{GR} imply
that there is at least
one weak vertex $s$ that satisfies that conditions
of Lemma~\ref{GR}.
Therefore, $s$ is a special vertex for $T_s$
such that {\sc Cycle Finder} with $s$ as input
detects a cycle with probability less than $2/3$.
By Lemma~\ref{lem-special}, the induced subgraph 
$G_{T_s}$ must be $\eps/4$-close (since we provided
an input of $\eps/4$ to {\sc Cycle-freeness tester}).
We set $U_i$ to be $T_s$, and $U_{i+1}$ to be the
complement of $\bigcup_{j \leq i} U_j$.
Repeating this procedure till it ends, we construct
the desired partition. This proves that 
$G$ is $\eps$-close to being cycle-free.
\qed
\end{proof}

It is easy to see that {\sc Cycle-freeness Tester}
always accepts forests and provides a certificate
of size $\poly(\eps^{-1}\log n)$ when rejecting.
The running time is $O(\sqrt{n}\poly(d\eps^{-1}\log n))$.
This completes the proof of Theorem~\ref{thm-main}.

\bibliographystyle{plain}
\bibliography{cycle_freeness}

\end{document}